\documentclass{article}

\usepackage{amsmath,amssymb,amsthm,fullpage,mathrsfs,pgf,tikz,caption,subcaption,mathtools,mathabx}
\usepackage[ruled,vlined,noresetcount]{algorithm2e}
\usepackage{amsmath,amssymb,amsthm,mathtools}
\usepackage{thmtools}
\usepackage[utf8]{inputenc} 
\usepackage[T1]{fontenc}    
\usepackage{hyperref}       
\usepackage{url}            
\usepackage{booktabs}       
\usepackage{amsfonts}       
\usepackage{nicefrac}       
\usepackage{microtype}      
\usepackage{times}
\usepackage{bbm}
\usepackage{enumitem}
\usepackage{xcolor}
\usepackage{cleveref}
\usepackage{bm}
\usepackage{float}
\usepackage[margin=1in]{geometry}
\usepackage{listings}
\usepackage{cite}

\newtheorem{theorem}{Theorem}[section]

\newtheorem{definition}[theorem]{Definition}
\newtheorem{conjecture}[theorem]{Conjecture}

\newtheorem{remark}[theorem]{Remark}

\newtheorem{claim}[theorem]{Claim}

\newcommand{\FF}{\mathbb{F}}
\newcommand{\RR}{\mathbb{R}}
\newcommand{\ZZ}{\mathbb{Z}}

\newcommand{\cE}{\mathcal{E}}
\newcommand{\cF}{\mathcal{F}}
\newcommand{\cL}{\mathcal{L}}

\newcommand{\interior}[1]{%
  {\kern0pt#1}^{\mathrm{o}}%
}

\renewcommand{\>}{\rangle}

\let\ker\relax
\let\im\relax

\DeclareMathOperator{\ker}{ker}
\DeclareMathOperator{\im}{im}

\DeclareMathOperator{\res}{res}

\DeclareMathOperator{\Mat}{Mat}

\renewcommand{\int}{\textnormal{int}}

\def\F2{\mathbb{F}_2}

\let\a\relax
\let\b\relax
\let\c\relax
\let\d\relax
\def\a{\alpha}
\def\b{\beta}
\def\c{\gamma}
\def\d{\delta}

\title{Proposals for 3D self-correcting quantum memory}
\author{Ting-Chun Lin\thanks{Department of Physics, University of California San Diego, CA, and Hon Hai Research Institute, Taipei, Taiwan. Email: \texttt{til022@ucsd.edu}.} \and Hsin-Po Wang\thanks{Department of Electrical Engineering and Graduate Institute of Communication Engineering, National Taiwan University, Taipei, Taiwan. Email: \texttt{hsinpo@ntu.edu.tw}.}\and Min-Hsiu Hsieh\thanks{Hon Hai Research Institute, Taipei, Taiwan. Email: \texttt{min-hsiu.hsieh@foxconn.com}.}}

\begin{document}

\sloppy

\maketitle

\begin{abstract}
  A self-correcting quantum memory
    is a type of quantum error correcting code
    that can correct errors passively through cooling.
  A major open question in the field is
    whether self-correcting quantum memories can exist in 3D.
  In this work, we propose two candidate constructions for 3D self-correcting quantum memories.
  The first construction is an extension of Haah's code, which retains translation invariance.
  The second construction is based on fractals with greater flexibility in its design.
  Additionally, we review existing 3D quantum codes and
    suggest that they are not self-correcting.
\end{abstract}

\section{Introduction}
A self-correcting quantum memory \cite{eczoo_self_correct} is a type of quantum error correcting code
  implemented through a Hamiltonian.
When this Hamiltonian is coupled to a sufficiently cold heat bath,
  the system is able to preserve coherent quantum information for a long duration.
This duration is called the memory time $T_{mem}$.

Self-correcting quantum memory correct errors in a fundamentally different way
  compared to traditional quantum error correcting codes.
In a traditional quantum error correcting code,
  errors are corrected through active interventions:
  syndromes are measured,
  errors are identified,
  and appropriate X or Z flips are applied to eliminate those errors.
In contrast, a self-correcting quantum memory
  corrects errors passively by leveraging its energy landscape:
  configurations with more errors have higher energies.
When the Hamiltonian is coupled to a cold bath,
  the system tends to evolve towards low-energy configurations,
  which corresponds to configurations with fewer errors.
As a result, the errors are automatically corrected
  without the need for active intervention.

The ability to correct errors passively
  gives self-correcting quantum memories an advantage over other quantum codes,
  especially for long-term storage.
This is similar to the use of magnetic tapes and hard drives for storing classical information,
  where data is stored with the intention to retrieve it several years later.
The benefit of these passive memories is their low energy consumption,
  as energy is only expended during the storage and retrieval of the information.
In contrast, traditional quantum error correcting codes
  require continuous active maintenance to preserve the information,
  resulting in a constant energy cost throughout the storage period.

But do self-correcting quantum memory exist at all?
It is known that 4D toric codes are self-correcting \cite{alicki2010thermal}
  and all 2D stabilizer codes are not self-correcting \cite{bravyi2009no,alicki2009thermalization}.
That means if our universe is 4D, we have a way to build self-correcting memories;
  and if our universe is 2D, we have no hope of building self-correcting memories within the framework of stabilizer codes.
However, since we live in a 3D universe,
  the question remains open and critical for the future of quantum technology:
Is there a quantum stabilizer code that can be realized as a 3D local Hamiltonian
  and exhibits self-correcting properties?

This question has been explored by many physicists,
  but significant challenges remain.
It became clear that 3D Hamiltonians based on topological quantum field theories (TQFTs)
  cannot be self-correcting
  due to the existence of string logical operator,
  which result in $T_{mem} \lesssim e^{a \beta - b}$.
However, for a long period, there were no models beyond TQFT.
A major breakthrough came in 2011 when Haah constructed the first 3D stabilizer code
  that is beyond the framework of TQFTs \cite{haah2011local}.
  The key feature of the code is that it is free of string logical operators.
  Haah discovered the model by searching through translation invariant codes that are free of string logical operators.
It has been observed in \cite{bravyi2011analytic} that the memory time of Haah's code grows as
  $T_{mem} \lesssim L^{3 \beta - 11}$
  as long as $L \le L_* \sim e^{\beta - 1}$.
  This implies that the maximum possible $T_{mem}$ scales as $e^{\Theta(\beta^2)}$.
However, the time is still dependent on $\beta$.
This leaves an open question:
  can we construct a family of codes of increasing size
  with $T_{mem} \to \infty$ for fixed $\beta$?
In this paper, this property is the formal definition for a code to be self-correcting.

A concept closely related to self-correction is the energy barrier.
In certain settings, it was observed that
  $T_{mem} \sim e^{\beta \cE}$,
  where $\cE$ is the energy barrier of the code.
This behavior is known as the Arrhenius law.
Consequently, much effort has been devoted to finding codes with large energy barriers.
It is known that TQFT models have a constant energy barrier, $\cE = \Theta(1)$.
Haah's code has a logarithmic energy barrier, $\cE = \Theta(\log L)$.
Later, Michnicki introduced the welded code \cite{michnicki20123} with a polynomial energy barrier, $\cE = \Theta(L^{2/3})$.
Recently, a family of works \cite{portnoy2023local,lin2023geometrically,williamson2023layer,li2024transform} achieved the optimal linear energy barrier, $\cE = \Theta(L)$.

Unfortunately, despite the large energy barriers,
  these codes likely do not satisfy $T_{mem} \to \infty$ for fixed $\beta$.
This is largely due to the issue that these codes are obtained by
  gluing large pieces of surface codes.
Therefore, these codes inherit the undesireable property of the surface code.
Further discussion can be found in \Cref{sec:non-construction-recent-codes}.

Another attempt to construct self-correcting quantum stabilizer code
  was proposed by Brell \cite{brell2016proposal}.
The idea is inspired by the fact that an Ising model on a Sierpiński carpet is known to be self-correcting \cite{vezzani2003spontaneous}.
Brell suggested that the hypergraph product of two such classical codes
  could potentially be self-correcting.
Furthermore, the Hausdorff dimension of this product is $2 + \epsilon$,
  which is less than $3$,
  making it a candidate for a 3D self-correcting quantum code.
However, as we will discuss in \Cref{sec:non-construction-brell},
  this construction is likely not self-correcting.

In this paper, we present two new attempts for constructing self-correcting quantum code in 3D.
The first construction in \Cref{sec:construction-polynomial} is an extension of Haah's code, which is also translation invariant.
These codes are relatively simple,
  likely have better parameters,
  and may be more feasible for physical implementation.
The second construction in \Cref{sec:construction-fractal} draws on ideas from fractals,
  aiming for a construction with greater flexibility,
  which might allow a mathematical proof of self-correction.
While neither proposal is conclusive,
  we hope these constructions will inspire future works towards proving the existence of self-correcting codes in 3D.

We also include a discussion that characterizes all geometrically local codes in \Cref{sec:characterization}.
Essentially, every 3D quantum CSS local code can be viewed as
  multiple 2D surface codes stacked along the squares of the integer lattice grid,
  with carefully chosen condensation data along the edges.
This perspective aligns with the framework of topological defect networks \cite{aasen2020topological,song2023topological}.

\section{Preliminary}

\subsection{Classical Error Correcting Codes}

A classical code is a $k$-dimensional linear subspace $C \subseteq \FF_q^n$
  which is specified by a parity-check matrix $H: \FF_q^n \to \FF_q^m$
  where $C = \ker H$.
$n$ is called the size and $k$ is called the dimension.
The distance $d$ is the minimum Hamming weight of a nontrivial codeword
\begin{equation}
  d = \min_{c \in C - \{0\}} |c|.
\end{equation}

The energy barrier $E$ is the minimum energy required to generate a nontrivial codeword by flipping the bits one at a time.
More precisely, given a vector $c \in \FF_q^n$,
  $|H c|$ is the number of violated checks.
In the physical context, the energy of a state is proportional to the number of violated checks,
  so we will refer to $\epsilon(c) = |H c|$ as the energy of vector $c$.
We say a sequence of vectors $\gamma_{a\to b} = (c_0=a, c_1, ..., c_t=b)$ is a walk from $a$ to $b$
  if $c_i, c_{i+1}$ differ by exactly one bit $|c_i - c_{i+1}| = 1$.
The energy of a walk $\epsilon(\gamma) = \max_{c_i \in \gamma} \epsilon(c_i)$ is defined as the maximum energy reached among the vectors in $\gamma$.
Finally, the energy barrier $\cE$ is defined as the minimum energy among all walks $\gamma$ from $0$ to a nontrivial codeword
\begin{equation}
  \cE = \min_{\gamma_{0\to c}, c \in C - \{0\}} \epsilon(\gamma_{0\to c}).
\end{equation}
Intuitively, the energy barrier is another way to characterize the difficulty for a logical error to occur, other than distance.

We say the classical code is a low-density parity-check (LDPC) code if each check interacts with a bounded number of bits, i.e. $H$ has a bounded number of nonzero entries in each row.

\subsection{Quantum CSS Codes}

A quantum CSS code is specified by two classical codes $C_x, C_z$
  represented by their parity-check matrices $H_x: \FF_q^n \to \FF_q^{m_z}, H_z: \FF_q^n \to \FF_q^{m_x}$
  which satisfy $H_x H_z^T = 0$.
$n$, $m_x$, and $m_z$ are the number of qubits, the number of $X$ and $Z$ checks (i.e. stabilizer generators), respectively.
The code consists of $X$ and $Z$ logical operators represented by $C_x$ and $C_z$,
  $X$ and $Z$ stabilizers represented by $C_z^\perp$ and $C_x^\perp$,
  and nontrivial $X$ and $Z$ logical operators represented by $C_x - C_z^\perp$ and $C_z - C_x^\perp$.
The dimension $k = \dim C_x - \dim C_z^\perp$ is the number of logical qubits.
The distance is $d = \min(d_x, d_z)$ where
\begin{equation}
  d_x = \min_{c_x \in C_x - C_z^\perp} |c_x|,\qquad
  d_z = \min_{c_z \in C_z - C_x^\perp} |c_z|
\end{equation}
  are called the $X$ and $Z$ distance
  which are the minimal weights of the nontrivial $X$ and $Z$ logical operators.
The energy barrier is $\cE = \min(\cE_x, \cE_z)$
  where $\cE_x$ and $\cE_z$ are the minimum energy required to create a nontrivial $X$ and $Z$ logical operator
\begin{equation}
  \cE_x = \min_{\gamma_{0\to c_x}, c_x \in C_x - C_z^\perp} \epsilon_x(\gamma_{0\to c_x}),\qquad
  \cE_z = \min_{\gamma_{0\to c_z}, c_z \in C_z - C_x^\perp} \epsilon_z(\gamma_{0\to c_z}),
\end{equation}
  where $\epsilon_x(c_x) = |H_x c_x|$ is the number of violated $Z$-checks of the $X$ Pauli operator $c_x$
  and $\epsilon_z(c_z) = |H_z c_z|$ is the number of violated $X$-checks of the $Z$ Pauli operator $c_z$.

We say the quantum code is a low-density parity-check (LDPC) code if each check interacts with a bounded number of qubits and each qubit interacts with a bounded number of checks, i.e. $H_x$ and $H_z$ have a bounded number of nonzero entries in each column and row.

We remark that a quantum CSS code naturally corresponds to a chain complex.
In particular, $H_x$ and $H_z$ induce the chain complex
\begin{equation}
  X: \FF_q^{m_x} \xrightarrow{\delta_0 = H_z^T} \FF_q^n \xrightarrow{\delta_1 = H_x} \FF_q^{m_z}.
\end{equation}
In reverse, a chain complex also defines a quantum CSS code.

\subsection{Local embedding of codes}

Given a code with $n$ bits (or qubits) and $m$ checks.
We say the code has a local embedding to a region $B \subseteq \RR^D$,
  if there exists a map from the bits and checks to the the region,
  $I: [n] \sqcup [m] \to B$,
  such that
  (1) the interaction is local, meaning that the Euclidean distance between each check and the qubit it interacts with is at most a constant,
  (2) the number of bits and checks in each unit sphere is at most a constant.
Additionally, we say a code is local in $D$ dimension if the code has a local embedding to $\RR^D$.

\subsection{Memory time}

The straightforward definition of memory time can be somewhat cumbersome.
This leads to the different notions of memory time.
  Some of which are mathematically different,
  while others are equivalent.
Since this paper does not delve into rigorous mathematical proofs involving memory time,
  we will keep the discussion concise
  and provide only one definition.

We first define memory time for classical codes.
To do so, we must first model how a system evolves when coupled to a heat bath at inverse temperature $\beta$.
In our context, the coupling is assumed to be local,
  and the system's evolution can be described by Glauber dynamics.
Glauber dynamics is a Markov process
  that describes how a state transitions to a new state.
Let $\tilde c(t): \FF_q^n \to \RR$ be a random variable over the words at time $t$.
Initially, at $t = 0$, $\tilde c(0)$ is set to be a codeword.
As time progresses, the state evolves,
  and as $t\to \infty$, $\tilde c(t)$ approaches the Gibbs distribution
  where $\Pr[\tilde c(t) = c] \,\propto\, e^{-\beta |H c|}$.

A choice of an evolution with this property is
\begin{equation}
  \frac{d}{dt} \Pr[\tilde c(t) = c] = \sum_{c' \in \FF_q^n} \Gamma(c' \to c) \Pr[\tilde c(t) = c'] - \sum_{c' \in \FF_q^n} \Gamma(c \to c') \Pr[\tilde c(t) = c]
\end{equation}
where $\Gamma(c' \to c) = \frac{e^{-\beta |H c|}}{e^{-\beta |H c'|} + e^{-\beta |H c|}}$.
It is straightforward to verify that $\Pr[\tilde c(t) = c] \,\propto\, e^{-\beta |H c|}$
  is a stable distribution under this evolution.

Memory time $T_{mem}$ is defined as the time
  when we can no longer reliably decode the stored information.
Let $D: \FF_q^n \to \FF_q^n$ be a decoder
  such that $D(c + z) = D(c) + z$, where $z$ is a codeword.
This condition implies that the decoder only depends on the syndrome.
We define $T_{mem}$ as the time such that
\begin{equation}
  T_{mem} = \sup_{t \ge 0} (\Pr[D(\tilde c(t)) = 0] \ge 2/3)
\end{equation}
with the initial condition $\tilde c(0) = 0$.
This expression says that $T_{mem}$ is the latest time
  at which the decoder can successfully recover the correct information
  at least $2/3$ of the time.
Note that definition of memory time $T_{mem}$
  depends on the code $C$, the decoder $D$, and the inverse temperature $\beta$.
However, for simplicity,
  the dependence on the decoder is often omitted.
When we want to explicitly show these dependencies,
  we write $T_{mem}(C, D, \beta)$.

We say that a family of classical codes $\{C_i\}$ (with decodes $\{D_i\}$) is self-correcting
  if there exists a constant $\beta > 0$,
  such that $T_{mem}(C_i, D_i, \beta) \to \infty$ as $i \to \infty$.

\section{Construction 1: based on polynomial} \label{sec:construction-polynomial}

In a series of works (see thesis \cite{haah2013lattice}),
  Haah introduced a general framework for describing translation-invariant quantum codes.
He then focused on a particular family in this framework
  to search for codes without logical string operators.
However, it is believed that this family does not include any self-correcting codes.
Therefore, we introduce a larger family of codes,
  which we believe does contain self-correcting codes.

\subsection{Polynomial formalism for translation-invariant codes}

Translation-invariant codes can be succinctly described through polynomials.
Consider a scenario where a qubit is placed at each integer lattice point in $\ZZ^3$.
We associate each lattice point $(i,j,k)\in \ZZ^3$
  with the monomial $x^i y^j z^k$.
We can then represent a check that acts on qubits in $S$ by the Laurent polynomial
  $f = \sum_{(i,j,k) \in S} x^i y^j z^k \in \FF_2[x, x^{-1}, y, y^{-1}, z, z^{-1}]$.

For a translation-invariant code,
  if $f$ is a check,
  then shifting it by one unit along the X axis is also a check,
  which corresponds to $xf$.
More generally, if $f$ is a check,
  then $rf$ is also a check for any monomial $r$.
This means that, although a translation-invariant code
  may have many checks,
  it is sufficient to provide a generating set of Laurent polynomials,
  which is typically of finite size.\footnote{
    This succinct representation is analogous to the generator polynomial in the study of cyclic codes.
}

In general, we can place more than one qubit at each lattice site.
This means each check is now associated with a vector of Laurent polynomials $(f_1, f_2, ...)$.
  Here, the check acts on the first qubits at locations $S_1$, the second qubits at locations $S_2$, and so on,
  where $f_a = \sum_{(i,j,k) \in S_a} x^i y^j z^k$.
For example, in this language, the toric code is described by
\begin{equation}
  h_X = \begin{pmatrix} 1+x^{-1} \\ 1+y^{-1} \end{pmatrix}, \qquad
  h_Z = \begin{pmatrix} 1+y \\ 1+x \end{pmatrix}.
\end{equation}
See \Cref{fig:toric-code}.

\begin{figure}
  \centering
  \includegraphics[width=0.6\linewidth]{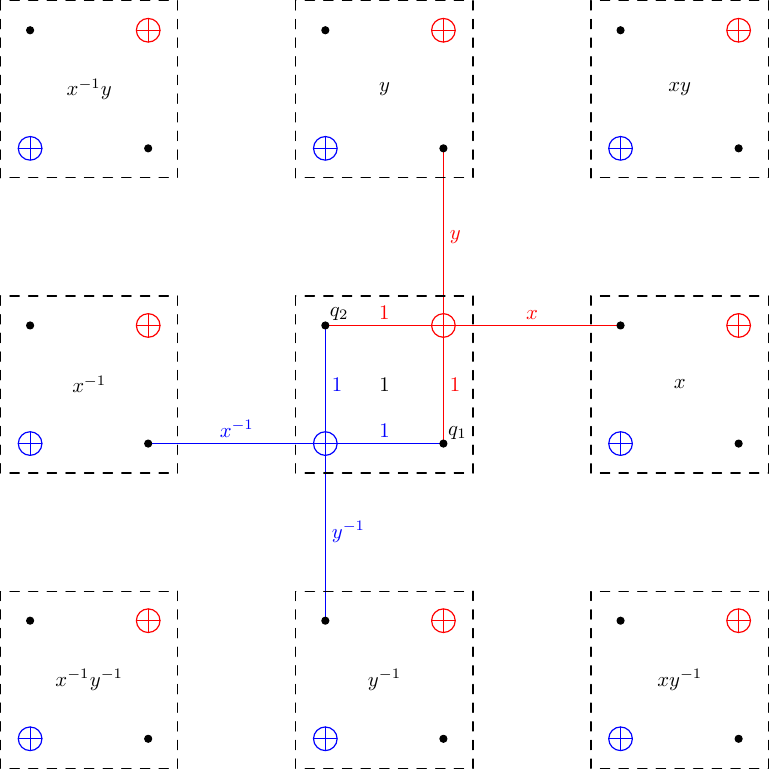}
  \caption{Each site contains one X-check (blue), two qubits (black), and one Z-check (red).
            The bottom right qubit at each site is the first qubit
              and the top left qubit at each site is the second qubit.
            We see that the X-check in the picture acts on
            the first qubit at $1$, the second qubit at $1$,
            the first qubit at $x^{-1}$, and the second qubit at $y^{-1}$.
            Similarly, the Z-check in the picture acts on
            the first qubit at $1$, the second qubit at $1$,
            the first qubit at $y$, and the second qubit at $x$.
            Therefore, $h_X = \big(\begin{smallmatrix} 1+x^{-1} \\ 1+y^{-1} \end{smallmatrix}\big)$
              and $h_Z = \big(\begin{smallmatrix} 1+y \\ 1+x \end{smallmatrix}\big)$.}
  \label{fig:toric-code}
\end{figure}

Recall that quantum CSS codes must satisfy the condition $H_X^T H_Z = 0$.
This implies that for any X check $(f_1, f_2, ...)$ and Z check $(f'_1, f'_2, ...)$,
  the total number of monomials shared between $f_a$ and $f'_a$
  across all $a$ must be even.
In the toric code example,
  the checks $h_X$ and $h_Z$ both contain the monomial $1$ for the first set of qubits and $1$ for the second set of qubits.
Furthermore,
  the checks $x h_X$ and $h_Z$ both contain the monomial $1$ for the first set of qubits and $x$ for the second set of qubits.
In fact, the condition on the number of shared monomials
  can be succinctly captured by the expression $\bar h_X^T h_Z = 0$,
  where $\bar f$ denotes the conjugate Laurent polynomial defined as $\bar f(x, y, z) = f(x^{-1}, y^{-1}, z^{-1})$.

We summarize the discussion above with the following definition of translation-invariant quantum codes.
\begin{definition}
  A translation-invariant quantum CSS code in D dimension over $\FF_q$ is specified by two matrices
  $h_X \in \Mat(n, m_x, R), h_Z \in \Mat(n, m_z, R)$,
  where $R = \FF_q[x_1^{\pm 1}, ..., x_D^{\pm 1}]$,
  such that $\bar h_X^T h_Z = 0$.
\end{definition}
This framework can also be applied to classical codes.
\begin{definition}
  A translation-invariant classical code in D dimension over $\FF_q$ is specified by a matrix
  $h \in \Mat(n, m, R)$,
  where $R = \FF_q[x_1^{\pm 1}, ..., x_D^{\pm 1}]$.
\end{definition}

\subsection{A quantum code family considered by Haah}

Satisfying the commuting condition
  $\bar h_X^T h_Z = 0$ is often the primary challenge in constructing quantum codes.
The goal is to impose enough structure on $h_X, h_Z$ to satisfy the commuting condition,
  while having sufficient variety to generate nontrivial properties.
A 3D local quantum code family that satisfies the commuting condition,
  while having more variety than the toric code,
  can be constructed by selecting two arbitrary Laurent polynomials
  $f, g \in \FF_q[x^{\pm 1}, y^{\pm 1}, z^{\pm 1}]$
  and define
\begin{equation}
  \bar h_X^T = \begin{pmatrix} f & g \end{pmatrix}, \qquad
  h_Z = \begin{pmatrix} g \\ -f \end{pmatrix}.
\end{equation}
This construction satisfies $\bar h_X^T h_Z = 0$ because $fg - gf = 0$,
  which can be illustrated by the following diagram.

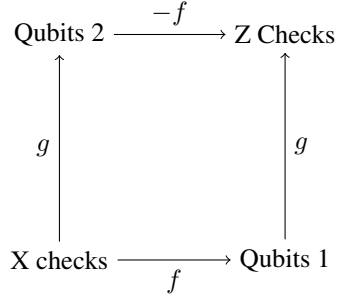
\begin{figure}[H]
  \centering
  \begin{tikzpicture}[scale=3]
    \draw (0,0)node[name=X]{X checks}
          (1,0)node[name=Q1]{Qubits 1}
          (0,1)node[name=Q2]{Qubits 2}
          (1,1)node[name=Z]{Z Checks};
    \draw[->] (X) --node[auto,swap]{$f$} (Q1);
    \draw[->] (X) --node[auto]{$g$} (Q2);
    \draw[->] (Q1) --node[auto,swap]{$g$} (Z);
    \draw[->] (Q2) --node[auto]{$-f$} (Z);
  \end{tikzpicture}
  \caption{A pictorial representation of a family of solutions.}
  \label{fig:Haah-code}
\end{figure}

By searching through models where $f, g$ are linear combinations of $1, x, y, z, xy, yz, zx, xyz$,
  Haah found a list of models without string operators (beyond TQFT) \cite{haah2011local}.
The most well-known among these has
  $f = 1 + x + y + z$ and $g = 1 + xy + yz + zx$.
However, as mentioned in the introduction, this model is not a quantum memory.
Therefore, we need to explore a broader family of codes to find candidates for quantum memories.

One generalization is to consider $f, g$ with higher degree polynomials.
However, this approach does not yield self-correcting codes
  due to the existence of self-similar patterns,
  which will be discussed in \Cref{sec:m-1}.
Instead, we explore a more complex parametrization for $h_X$ and $h_Z$
  that still satisfies $\bar h_X^T h_Z = 0$.

\subsection{A new quantum code family}

If we examine why $\bar h_X^T h_Z = 0$ in the above code family,
  the key reason is due to the square in \Cref{fig:Haah-code}.
So a natural generalization is to consider models with more squares,
  for example, by taking the Cartesian product of two complete bipartite graphs.
We can then assign functions to the edges, which induce well-defined families of quantum codes.

Let $I_1, J_1$ be vertex sets of size $m_1$
  and $I_2, J_2$ be vertex sets of size $m_2$.
We form two complete bipartite graphs $G_1, G_2$; one between $I_1, J_1$ and the other between $I_2, J_2$.
After the Cartesian product of $G_1$ and $G_2$,
  we obtain a 4-partite square complex $S$.

With the square complex $S$,
  we associate the vertices in $I_1 \times I_2$ with X-checks,
  the vertices in $(I_1 \times J_2 ) \cup (J_1 \times I_2)$ with qubits,
  and the vertices in $J_1 \times J_2$ with Z-checks.
So, there is a total of $m_1 m_2$ X-checks, $2 m_1 m_2$ qubits, and $m_1 m_2$ Z-checks per lattice site.

We now describe the two matrices $h_X, h_Z$.
We associate each edge in $G_1$, $(i_1, j_1) \in I_1 \times J_1$,
  with a function $f_{i_1, j_1}$
  and each edge in $G_2$, $(i_2, j_2) \in I_2 \times J_2$,
  with a function $g_{i_2, j_2}$.
These functions on the edges are naturally carried to the edges of the square complex.
For example, the edge $((i_1, i_2), (j_1, i_2))$ in the square complex $S$,
  is associated to the function $f_{i_1, j_1}$.
These functions induce $h_X, h_Z$ with appropriate conjugation and negation.
Overall, we have
\begin{align}
  (h_X)_{(i_1, i_2), (j_1, i_2)} &= \bar f_{i_1, j_1}, \\
  (h_X)_{(i_1, i_2), (i_1, j_2)} &= \bar g_{i_2, j_2}, \\
  (h_Z)_{(j_1, i_2), (j_1, j_2)} &= g_{i_2, j_2}, \\
  (h_Z)_{(i_1, j_2), (j_1, j_2)} &= -f_{i_1, j_1}.
\end{align}
It is straightforward to verify that $\bar h_X^T h_Z = 0$
  and the quantum code is well-defined.

For example, when $m_1 = 1$ and $m_2 = 1$,
  this recovers the family considered by Haah.
When $m_1 = 2$ and $m_2 = 2$, we have
\makeatletter
\renewcommand*\env@matrix[1][*\c@MaxMatrixCols c]{%
  \hskip -\arraycolsep
  \let\@ifnextchar\new@ifnextchar
  \array{#1}}
\makeatother
\begin{equation}
  \bar h_X^T =
    \begin{pmatrix}[cccc|cccc]
      f_{11} & 0 & f_{21} & 0 & g_{11} & g_{21} & 0 & 0 \\
      0 & f_{11} & 0 & f_{21} & g_{12} & g_{22} & 0 & 0 \\
      f_{12} & 0 & f_{22} & 0 & 0 & 0 & g_{11} & g_{21} \\
      0 & f_{12} & 0 & f_{22} & 0 & 0 & g_{12} & g_{22} \\
    \end{pmatrix},\qquad
  h_Z =
    \begin{pmatrix}
      g_{11} & g_{12} & 0 & 0 \\
      g_{21} & g_{22} & 0 & 0 \\
      0 & 0 & g_{11} & g_{12} \\
      0 & 0 & g_{21} & g_{22} \\ \hline
      -f_{11} & 0 & -f_{12} & 0 \\
      0 & -f_{11} & 0 & -f_{12} \\
      -f_{21} & 0 & -f_{22} & 0 \\
      0 & -f_{21} & 0 & -f_{22} \\
    \end{pmatrix},
\end{equation}
where the labels are ordered by $(1,1), (1,2), (2,1), (2,2)$.

We conjecture that within this family, there exists a code that saturates several quantum code upper bounds
  and has self-correcting property.
\begin{conjecture}
  For a suitable choice of integers $m_1, m_2 \ge 2$, a finite field $\FF_q$, and functions $f_{i_1, j_1}, g_{i_2, j_2}$,
    for $i_1, j_1 \in [m_1]$, $i_2, j_2 \in [m_2]$,
    which are linear combinations of $1, x, y, z, xy, yz, zx, xyz$ over $\FF_q$,
    the corresponding translation invariant 3D quantum code on a cube of size $L$,
    with appropriate boundary conditions,
    has parameters $[[n = \Theta(L^3), k = \Theta(L), d = \Theta(L^2), \cE = \Theta(L)]]$
    and a memory time of $\exp(\Theta(L))$.
\end{conjecture}

Empirically, for a large enough $q$, there seem to be a qualitative difference between
  codes with $m_1, m_2 = 1$ and $m_1, m_2 \ge 2$.
  This is what set us apart from previous constructions.
This dichotomy will be explore more in \Cref{sec:m-1,sec:m-2}.
Indeed, we suspect that $m_1, m_2 = 2$ is enough to find an instance of self-correcting quantum code, with large enough $q$.

Even though the conjecture is stated for qudits over $\FF_q$,
  when $q$ is a power of $2$,
  we can apply alphabet reduction
  and convert the code over $\FF_q$ into a code over $\FF_2$
  (see for example \cite{golowich2024asymptotically,nguyen2024good}).

\begin{remark}
  In the above description,
    we focused on the construction based on product of two complete bipartite graphs.
  More generally, one can replace the product structure with a ``labeled'' square complex,
    which satisfies the property that opposite edges of every square have the same label.
  We can then associate each horizontal edge with label $a \in A$ with a function $f_a$,
    and each vertical edge with label $b \in B$ with a function $g_b$.
\end{remark}

\begin{remark}
  In \cite[Theorem 4.3]{haah2013lattice},
    Haah discussed a type of no-go result for 3D translation-invariant codes.
  Specifically,
    the theorem states that
    any 3D translation-invariant code
    that is exact (i.e. has no nontrivial local logical operator)
    and degenerate (i.e. $k > 0$ for some torus of finite size $L$)
    have fractal generators.
  These fractal generators can proliferate,
    which likely implies that such a code cannot be self-correcting.
  We note that our construction may potentially circumvent this barrier
    by having $k = 0$ for every torus of size $L$,
    while allowing $k > 0$ for other types of boundary conditions.
\end{remark}

\subsection{A new classical code family}

This construction strategy can also be applied to 2D classical codes,
  where $h$ is a $m$ by $m$ matrix
  and each entry contains a linear function over $x, y$.
We have a similar conjecture regarding the properties of these classical codes.

\begin{conjecture}
  For a suitable choice of integer $m \ge 2$ and functions $f_{i,j}$, for $i, j \in [m]$, which are linear combinations of $1, x, y, xy$ over $\FF_q$,
    the corresponding translation invariant 2D classical code
    with $h_{i,j} = f_{i,j}$
    on a torus of width $L$
    has parameters $[n = \Theta(L^2), k = \Theta(L), d = \Theta(L^2), \cE = \Theta(L)]$
    and has memory time $\exp(\Theta(L))$.
  Additionally, the code induced from $h^T$ have the same properties.
\end{conjecture}

This conjecture is expected to be easier to solve than the quantum case.
This could be a good exercise before tackling the quantum version.

We assume the index sets $I, J$ have the same size,
  because we want both $h$ and $h^T$ to have good properties.
It is important for both directions $h$ and $h^T$ to have good properties
  because for quantum codes both directions matters.
If the two index sets $I, J$ have different sizes,
  then one of the two codes will have a larger bits density than the check density.
If we further assume that the checks are independent,
  then the codes will have constant size codewords,
  i.e. constant distance,
  which is not good.
Thus, we assume $I, J$ have the same sizes.

Note that when $m = 1$, this corresponds to the code studied by Yoshida \cite{yoshida2013information}.
However, our new code family is conjectured to be better in the parameters $d$ and $\cE$
  which are only $\Theta(L^{2-\epsilon})$ and $\Theta(\log(L))$ when $m = 1$.
Additionally, the code family with $m = 1$ is not self correcting.

Our new code family is also conjectured to be better than the classical codes considered in
  \cite{lin2023geometrically,baspin2023combinatorial}.
One aspect is that this new family is translation invariant.
The other aspect is that the previous family is not self correcting,
  which follows from the discussion in \Cref{sec:non-construction-recent-codes}.

\subsection{Upper bounds for codes with $m = 1$} \label{sec:m-1}

The difference between $m = 1$ (or $m_1, m_2 = 1$ for quantum codes)
  versus $m \ge 2$ (or $m_1, m_2 \ge 2$ for quantum codes)
  is the main distinction between our proposal and previous constructions.
It is generally believed that
  2D classical codes with $m = 1$ and 3D quantum codes with $m_1, m_2 = 1$
  only have memory time $L^{\Theta(\beta)}$ when $L < L_* = e^{\Theta(\beta)}$
  and do not increase further $L_*^{\Theta(\beta)}$ when $L > L_*$.
We will explain some of the ideas behind this belief in this section.
We will focus the discussion on classical codes for simplicity,
  but the discussion also applies to quantum codes.

When $m = 1$, the 2D classical code is specified by a function $f = \a + \b x + \c y + \d xy$.
We will construct words with the following property.
A similar discussion can be found in the paragraphs below \cite[Def 3.10]{haah2013lattice}.
\begin{claim}
  Given a $m = 1$ 2D translation invariant code
    constructed from $f = \a + \b x + \c y + \d xy$.
  There exist words $c_\ell$ with increasing weight
    such that the syndrome $f c_\ell$ has weight $\le 4$.
  Furthermore, there exist a path to $c_\ell$ with energy $\le \Theta(\log |c_\ell|)$.
\end{claim}

\begin{proof}
  The construction of these words utilizes the trick
  \begin{equation}
    (x+y)^p = \sum_{i=0}^p {p \choose i} x^i y^{p-i} = x^p + y^p
  \end{equation}
    where $\FF_q$ has characteristic $p$.
  Let $c_\ell = f^{p^\ell - 1}$ be a family of words.
  The corresponding syndrome for $c_\ell$ is $f c_\ell = f^{p^\ell}$.
  Using the trick, we have
  \begin{equation}
    f^{p^\ell}
    = (\a + \b x + \c y + \d xy)^{p^\ell}
    = \a^{p^\ell} + \b^{p^\ell} x^{p^\ell} + \c^{p^\ell} y^{p^\ell} + \d^{p^\ell} x^{p^\ell}y^{p^\ell}.
  \end{equation}
  This means that the syndrome has weight $\le 4$ which can only violates the checks at
    $(0,0), (p^\ell, 0), (0, p^\ell), (p^\ell, p^\ell)$.

  Meanwhile, $c_\ell$ generally has a large support
    due to a self-similar structure.
  Notice that
    \begin{equation}
      c_{\ell+1} = f^{p^{\ell+1} - 1} = \big(f^{p-1}\big)^{p^\ell} f^{p^\ell - 1} = b^{p^\ell} c_\ell
    \end{equation}
    where $b$ is defined as $f^{p-1} = \sum_{i, j=0}^{p-1} b_{i, j} x^i y^j$.
  Following the same argument,
  \begin{equation}
    \big(f^{p-1}\big)^{p^\ell} = \sum_{i, j=0}^{p-1} b_{i, j}^{p^\ell} x^{i p^\ell} y^{j p^\ell}
  \end{equation}
  which is only supported on $(i p^\ell, j p^\ell)$ for $0 \le i, j \le p-1$
    with weight $p^2$.
  This implies that $c_{\ell+1}$ can be built by
    adding $p^2$ copies of $c_\ell$,
    whose $(i p + j)$-th copy is
    multiplied by $b_{i, j}^{p^\ell}$
    and shifted to $(i p^\ell, j p^\ell)$.
  This is the self-similar structure that appears when $m = 1$.

  Note that because $c_\ell$ is supported within $\{(i', j'): 0 \le i', j' \le p^\ell-1\}$,
    these copies do not overlap,
    which make it easy to compute the size of the support.
  In particular,
  \begin{equation}
    |c_{\ell+1}| = A_n |c_\ell|
  \end{equation}
  where $A_n$ is the number of nonzero values in $\{b_{i, j}^{p^\ell}\}_{0 \le i, j \le p-1}$.
  Since $b_{i, j}^{p^\ell}$ is nonzero iff $b_{i, j}$ is nonzero,
    it is clear that $A_n = A_0$.
  By applying the argument iteratively,
    we have
    \begin{equation}
      |c_\ell| = A_0^\ell.
    \end{equation}
  Unless the code is degenerate, where at most one of $\a, \b, \c, \d$ is nonzero,
    $A_0 \ge 2$.
  Hence $c_\ell$ has increasing weight.
  (In fact, the example to have in mind should have at least $3$ nonzero values among $\a, \b, \c, \d$.
    Otherwise, the code has string operators.)

  Finally, we want to show a path to $c_\ell$ with energy $\le \Theta(\ell) = \Theta(\log |c_\ell|)$.
  Because of the self-similar structure,
    $c_\ell$ can be constructed by
    building $A_0$ copies of $c_{\ell-1}$ with some multiplicative factor.
  This leads to a tree structure where each node has $A_0$ subtrees.
  We can then build $c_\ell$ through the order taken by depth first search.
  In another words, we build each $c_{\ell-1}$ in order.
  And to build a $c_{\ell-1}$,
    we build each $c_{\ell-2}$ within in order.
  That means at any given time,
    we are in the process of building some $c_{\ell-1}$,
      which is in the process of building some $c_{\ell-2}$ inside such $c_{\ell-1}$
      and so on.
  This corresponds to a particular leaf.

  We are ready to bound the energy of such path.
  A node whose subtrees have been explored corresponds to a completed $c_{\ell'}$
    which only contributes at most $4$ units of energy.
  At any given moment,
    the explored leaf can be covered $\le \ell A_0$ completed nodes.
  That means the energy of the path is at most $4 \ell A_0 = \Theta(\ell) = \Theta(\log |c_\ell|)$.
\end{proof}

The claim suggests that there are large words of diameters $\Theta(L)$
  (because $c_\ell$ has large weight)
  that appear after a short period $e^{\beta \Theta(\log L)}$
  (because $c_\ell$ can be reached with low activation energy).
Indeed, these fractal patterns can be observed in the Monte Carlo simulations of Glauber dynamics.
As time progresses,
  these fractal patterns obscure the codewords and destroys the encoded information.
While some of these statements lack rigorous proof,
  they likely represent a correct physical picture of the underlying thermodynamic process.

\subsection{Discussions on $m \ge 2$} \label{sec:m-2}

As we saw in the last section,
  codes with $m = 1$ have large patterns with small syndrome that destroy the encoded information.
We conjecture that, for large enough $m \ge 2$,
  there is no longer large patterns with small syndrome.
In fact, we conjecture something stronger.
\begin{conjecture} \label{conj:isoperimetric}
  There exists a 2D translation invariant code $h$ and a constant $\lambda > 0$,
    such that for any word $c$ with finite support
    $|hc| \ge \lambda |c|^{1/2}$.
\end{conjecture}
This bound is analogous to the isoperimetric inequality in $\RR^2$,
  where there exists a constant $\lambda$,
  such that $|\partial A| \ge \lambda |A|^{1/2}$
  for any compact region $A \subset \RR^2$,
  where $\partial A$ is the boundary of $A$.
This is the key property that implies 2D Ising model is self correcting.

In the proof for 2D Ising model,
  the above isoperimetric property is then used to bound the number of low energy states.
\begin{conjecture}
  There exists a 2D translation invariant code $h$ and a constant $\kappa > 0$,
    such that for any $w$ and $L$,
    the number of irreducible words supported in $[0,L] \times [0,L]$
    with energy $\le w$,
    is $\le L^2 e^{\kappa w}$.
\end{conjecture}
We say a word $c$ is reducible if
  there exist $c', c''$
  such that $c = c' + c''$
  and the support of $f c'$ and the support of $f c''$ are disjoint.
Otherwise, we say $c$ is irreducible.

The $L^2$ factor comes from the fact that the code is translation invariant,
  which means if a word has energy $w$,
  then a translation of the word also has energy $w$.
Since the system has roughly $L^2$ translation directions,
  this leads to the $L^2$ factor.

One reason to study irreducible words is that
  this is the way the analysis is done
  in the proof regarding the thermal stability of the 2D Ising model.
A more physical reason is that each of these irreducible words
  are evolving independently
  and the analysis should know this feature.
In contrast, if we study the low energy states directly,
  then the bound on the number of low energy state
  do not capture the properties of the irreducible words,
  since there are significantly more reducible states than irreducible states.
For example, there are $\Theta(L^2)$ words of weight $1$,
  and there are $\Theta(L^4)$ words of weight $2$,
  most of which come from reducible words supported on two distant points.
More generally, since a word of weight $w' = w/4$ has energy at most $w$,
  there are at least $\Omega(L^{2w'}) = \Omega(L^{w/2})$ low energy words
  which is much larger than the number of irreducible words $L^2 e^{\kappa w}$.

We provide an example that sketches how the bound on the number of irreducible words is used.
Consider an extreme scenario where there are two disjoint regions $A_1, A_2$ in the code
  (this do not happen for generic 2D translation invariant codes).
It implies every word $c$ can be separated into a component on $A_1$, $c_1$, and another component on $A_2$, $c_2$,
  such that $c = c_1 + c_2$.
In this case, the partition function $\sum_{c} e^{-\beta c}$,
  can be regrouped as $\sum_{c} e^{-\beta c} = (\sum_{c_1} e^{-\beta c_1})(\sum_{c_2} e^{-\beta c_2})$.
A similar trick can be applied to the general case
  which separates the reducible words into irreducible components.
Roughly speaking, albeit being incorrect, we have
  $\sum_{c} e^{-\beta c} = \prod_{c_{irr}} (1 + e^{-\beta c_{irr}})$
  where $c_{irr}$ ranges over all irreducible components.
These steps are hidden in the proof regarding the thermal stability of 2D Ising model.

One might wonder why we remain optimistic that these conjectures can be satisfied
  by some 2D translation invariant code,
  even though they are not satisfied by $m = 1$ codes.
We do not have direct evidence.
Instead, we resort to analogy.

In the context of cellular automaton and periodic tiling,
  it is known that different instances have drastically different behaviors.
The simplest constructions often display simple patterns,
  such as periodic patterns.
More complex constructions may yield fractal patterns.
But the most complicated instances are known to be Turing complete.

We suggest that a similar dichotomy may exist for translation invariant codes.
It seems possible that by going to $m \ge 2$,
  we can circumvent the undesired fractal patterns.
At the very least, the construction of the fractal patterns described in \Cref{sec:m-1}
  no longer works.
More generally, we might hope that the allowed patterns are sufficiently generic,
  so that the condition in \Cref{conj:isoperimetric} holds.
We believe that most instances with large $m$ and large field size $q$
  should be self-correcting.
In particular, self-correcting codes are abundant in this regime of large $m$ and $q$.

We tried to demonstrate the existence of self-correcting codes
  by numerically running Monte Carlo and estimate the memory time for the classical 2D codes.
One signature for this dichotomy is that $T_{mem}$ scales polynomially versus exponentially with $L$.
However, the memory times are quite long,
  and we are only able to measure the memory time for small system sizes $L = 3, 4, 5$, and barely $6$.
Therefore, we are unable to confidently distinguish between these two scalings over $L$.

To summarize, before the paper by Haah and Yoshida,
  our understanding of codes with translation symmetry was limited to TQFTs,
    whose excitations are string-like or surface-like.
The papers by Haah and Yoshida,
  introduced a new regime, where the excitations are fractal-like.
In this paper, we suggest the existence of yet another regime,
  where the excitations are beyond strings or fractals,
  and potentially beyond simple classifications.

\section{Construction 2: based on fractals}
\label{sec:construction-fractal}

Even though we believe that Construction 1 yield an instance of self-correcting quantum code
  and may be more practical to physically realize,
  it may be difficult to prove its properties mathematically.
So we present an alternative route through a different construction.
This construction has worse parameters than Construction 1.
However, it does not need to be translation invariant,
  which provides more flexibility to the construction.

Construction 2 is based on the hypergraph product of two classical codes.
The two classical codes are supported on a Sierpiński carpet of Hausdorff dimension $1+\epsilon$
  and are local respect to the geometry of the Sierpiński carpet.
This means the quantum code is supported on a product of two Sierpiński carpets,
  which can be embedded in $\RR^3$ (will be described in another paper).
This constitutes the construction of a 3D local quantum code.
The missing piece is to show that the quantum code is self correcting.
We will discuss each component in more detail.

\subsection{Review of fractals and prefractals}

Fractals are point sets with self similar structures.
Some common examples include Cantor set, Sierpiński carpet, and Menger sponge.
We will focus our discussion to Sierpiński carpets.
Fractals are important in our context because they have fractional Hausdorff dimensions
  which allows us to explore regimes that was not possible with integer dimensions.

The Sierpiński carpet is a fractal subset of $[0,1]^2$ with self similar structures.
It is defined by cutting a square into $9$ subsquares
  and removing the center square.
The process is then applied recursively to the remaining $8$ subsquares.
The limiting set has Hausdorff dimension $\frac{\log 8}{\log 3} \approx 1.893$.

For the purpose of embedding a product of two fractals into $\RR^3$
  we want the fractals to have Hausdorff dimension $1 + \epsilon$ for some small $\epsilon > 0$.
This can be achieved by cutting the square into more subsquares.
Given an integer $A > 0$.
Let $C_0 = [0,1]^2$, which consists of one square.
Given $C_i$, which consists of several squares,
  we cut each square into $A \times A$ subsquares,
  and remove the center $(A-2) \times (A-2)$ subsquares.
The remaining squares form the region $C_{i+1}$,
  and this process is done recursively.
The limiting set is formally defined as the intersection $C = \bigcap_{i=0}^\infty C_i$
  which has Hausdorff dimension approaching $1$ as $A$ approaches infinity.

A related concept of fractals is prefractals
  which are the intermediate sets that appear during the iteration process.
For example, the sets $C_0, C_1, ...$ described above.
In a formal discussion, we will use prefractals instead of fractals,
  because while our code mimics a fractal,
  it has to be a finite object.
Therefore, we must terminate the process after a finite number of iterations.

Additionally, we rescale the sets so that each square has a unit size.
This rescaling ensures a constant qubit density,
  as we will later place a fixed number of qubits in each square.
In particular, let $B_i = A^i C_i$,
  which scales up $C_i$ by a factor of $A^i$.

\subsection{Classical code supported on a prefractal}

The classical code needed for our construction are required to have very good properties
  that are not known to exist yet.
These are similar to the isoperimetric inequiality described in \Cref{conj:isoperimetric}.

\begin{conjecture} \label{conj:isoperimetric-2}
  For any $A > 0$,
    there exist a classical code with the parity check matrix $H: \FF_q^n \to \FF_q^{n'}$
      and a local embedding $[n] \sqcup [n'] \to B_i$
    such that both $H: \FF_q^n \to \FF_q^{n'}$ and $H^T: \FF_q^{n'} \to \FF_q^n$ satisfy
    $|H c| \ge \lambda |c|^\nu$ and $|H^T c'| \ge \lambda |c'|^\nu$
    for any $c \in \FF_q^n, c' \in \FF_q^{n'}$ with $|c|, |c'| \le n^\mu$
    (note that $n' = \Theta(n)$ because of the local embedding),
    for some constant $\lambda, \mu, \nu > 0$.
\end{conjecture}

Because we describe the system as a finite size object,
  we additionally need the small set condition, $|c|, |c'| \le n^\mu$.
In words, the conjecture states
  that $H$ and $H^T$ have small set isoperimetric inequality.

It is plausible that such code exists
  through random constructions.
One way to generate a random code on the scaled prefractal $B_i$
  is by placing bits and checks on $B_i$ randomly
  and connect them randomly.
As we forgo translation invariance,
  the choice of interaction in each local region is independent from the rest.
This provides more flexibility in construction and proof techniques
  which may be sufficient to establish the isoperimetric inequality.\footnote{
    The simple construction described here is not the full story.
    In fact, if the connections are done randomly,
      there will always be degenerate structures,
      when the system is large enough.
    Therefore, we need to pass the random construction through the Lovász local lemma,
      to guarantee there is no bad events.
  }

On the other hand, we may view this problem through the language of linear circuits.
For example, a bit and a check can be interpreted as a gate that maps $(x, y)$ to $(x, x+y)$. See \Cref{fig:cnot}.
Therefore, the question becomes a question about building robust circuits that is supported on the prefractal $B_i$.

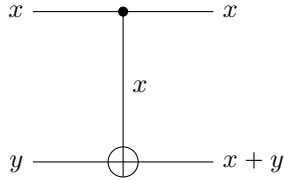
\begin{figure}[H]
  \centering
  \begin{tikzpicture}[scale=2]
    \draw (-0.6,1)node[anchor=east]{$x$} -- (0.6,1)node[anchor=west]{$x$};
    \draw (-0.6,0)node[anchor=east]{$\phantom{x+}y$} -- (0.6,0)node[anchor=west]{$x+y$};
    \draw (0,-0.1) -- (0,1);
    \draw (0,0.5)node[anchor=west]{$x$};
    \draw (0,0) circle (0.1cm);
    \filldraw (0,1) circle (0.03cm);
  \end{tikzpicture}
  \caption{The figure illustrates an example where we interpret a linear code as a linear circuit.}
  \label{fig:cnot}
\end{figure}

\subsection{The candidate quantum code}

As we saw, even the classical code problem is already unknown.
Nevertheless, we will hint on how the classical problem can help us construct self-correcting quantum codes.
The quantum code we consider is the hypergraph product of two classical codes supported on a prefractal $B_i$.
That means the quantum code is local on the Cartesian product $B_i \times B_i$.
This product naturally embeds in $\RR^4$.
But, in fact it can be embedded in $\RR^3$ with an appropriate choice of $A$.

\begin{theorem}
  For large enough $A$, $B_i \times B_i$ can be embedded in $\RR^3$.
\end{theorem}

The proof of the theorem will be demonstrated in a separate paper.
Notice that Sierpiński carpet with parameter $A$ has Hausdorff dimension $1+\epsilon$
  with $\epsilon \to 0$ as $A \to \infty$.
This implies the volume of $B_i \times B_i$ scales as the diameter to the power of $2(1+\epsilon)$.
When $\epsilon$ is small enough,
  this power is less than $3$,
  which suggests the possibility to embed $B_i \times B_i$ into $\RR^3$.

Indeed, the random embedding techniques used in \cite{gromov2012generalizations,portnoy2023local}
  is morally saying that the embedding problem boils down to a comparison of dimensions.
Although, these were studied in the context of integer dimensions,
  we are able to apply those techniques to the context of fractals with non-integer dimensions.

We have now established a family of 3D local quantum codes.
However, the challenge is to demonstrate that these codes satisfy the small-set isoperimetric inequality and are self-correcting.
Again, these problems are harder than the classical versions.
One should first solve \Cref{conj:isoperimetric-2} before attempting the following problems.

\begin{conjecture} \label{conj:isoperimetric-3}
  For any $A > 0$,
    there exist two classical codes with the parity check matrices $H_1: \FF_q^{n_1} \to \FF_q^{n_1'}$, $H_2: \FF_q^{n_2} \to \FF_q^{n_2'}$
      and local embeddings $[n_1] \sqcup [n_1'] \to B_i$, $[n_2] \sqcup [n_2'] \to B_i$.
  Let $X: \FF_q^{X(0)} \xrightarrow{\delta_0} \FF_q^{X(1)} \xrightarrow{\delta_1} \FF_q^{X(2)}$ be the tensor product of $H_1$ and $H_2$,
    where $X: \FF_q^{n_1} \otimes \FF_q^{n_2} \to
      (\FF_q^{n_1} \otimes \FF_q^{n_2'}) \oplus (\FF_q^{n_1'} \otimes \FF_q^{n_2}) \to
      \FF_q^{n_1'} \otimes \FF_q^{n_2'}$.
  Then, for any $c \in \FF_q^{X(1)}$ with $|c| \le n^\mu$,
    $|\delta_1 c| \ge \lambda (\min_{b \in \im \delta_0} |c + b|)^{\nu}$
    and $|\partial_1 c| \ge \lambda (\min_{b \in \im \partial_2} |c + b|)^{\nu}$,
    for some constant $\lambda, \mu, \nu > 0$.
\end{conjecture}

\begin{conjecture}
  Under the same conditions as \Cref{conj:isoperimetric-3},
    the quantum code $X$ is self-correcting.
\end{conjecture}

\section{A characterization of geometrically local codes} \label{sec:characterization}

In this section we describe a family of local codes
  such that for any geometrically local code,
  there is an ``equivalent'' code within this family,
  which has the same code properties up to a constant factor\footnote{
    Formally, the equivalence between LDPC codes can be defined as chain homotopy with sparsity constraint.
    See, \cite[Sec 2.1]{lin2024transversal}.
  }.
This can be thought of as the counterpart to the polynomial formalism discussed in \Cref{sec:construction-polynomial}
  which characterizes geometrically local translation-invariant codes.
The benefit of this description is that it provides insights into the space of geometrically local codes
  and can be used to define the notion of random geometrically local codes.
This formalism can be extended to codes with arbitrary locality constraints.
Broadly speaking, this perspective aligns with the literature on topological defect networks \cite{aasen2020topological,song2023topological}.

\subsection{Classical codes}

We describe the family of 2D classical local codes.
The notation follows closely with \cite{lin2024transversal}.
Let $Y$ be the cellular complex that subdivides $\RR^2$ using the integer lattice.
We denote $Y(i)$ as the set of $i$-cells
  and $Y_i$ as the $i$-skeleton.
We use the notation $v \prec e$ to indicate that a lower-dimensional cell $v$ contained in a higher-dimensional cell $e$.
All vector spaces below are defined over $\FF_q$
  and have chosen bases.

Each edge is associated with a vector space $\cF_e$.
Each vertex is associated with a vector subspace $\cF_v \subseteq \prod_{e \succ v} \cF_e$.
The embedding induces a restriction map $\res_{v, e}: \cF_v \to \cF_e$
  for each $v \prec e$,
  which projects to the component at $e$.

The global code $C \subseteq \prod_{e \in Y(1)} \cF_e$ is defined as
  the set of $c \in \prod_{e \in Y(1)} \cF_e$
  such that $\{c(e)\}_{e \succ v} \in \cF_v$.
The parity-check matrix of the code can be explicitly written as follows.
Let $\prod_{v \in Y(0)} \cF_v \xrightarrow{\delta} \prod_{e \in Y(1)} \cF_e$,
  where $\delta(c)(e) = \sum_{v \prec e} \res_{v, e} c(v)$
  for $c \in \prod_{v \in Y(0)} \cF_v$.
One can verify that $\delta^T$ is the parity-check matrix of $C$,
  where the transpose is applied according to the chosen basis of the vector spaces.
It is clear that the code $C$ is geometrically local.

To summarize,
  specifying the code $C$ requires
  a 1D cellular complex that embeds into the desired geometry,
  a vector space $\cF_e$ for each edge $e \in Y(1)$,
  and a vector subspace $\cF_v \subseteq \prod_{e \succ v} \cF_e$
    for each vertex $v \in Y(0)$.
Notice that the choice of basis for each vector space affects the code properties only by a constant factor,
  since the code is LDPC.

We now sketch how to transform any 2D geometrically local classical code
  into a code within the above family.
First, draw a segment between every bit and check with nontrivial interaction.
  This forms the Tanner graph of the code.
Next, slightly deform the graph by a constant distance
  so that every bit and check lies in $Y_0$,
  and the segments are supported on $Y_1$.
Note that a deformed segment may span multiple cells in $Y(1)$.
This allows us to define the dimension of $\cF_e$ based on the number of segments on top of $e$.
Finally, the edges around a vertex are constrained by the bit and check requirements
  which specifies $\cF_v$.
Note that at the end the bits are placed on the edges and the check are placed on the vertices.
Essentially, this means that every 2D classical geometrically local code
  can be viewed as a Tanner code over the 1D skeleton on the integer lattice grid.

\begin{figure}[H]
  \centering
  \includegraphics[width=0.7\linewidth]{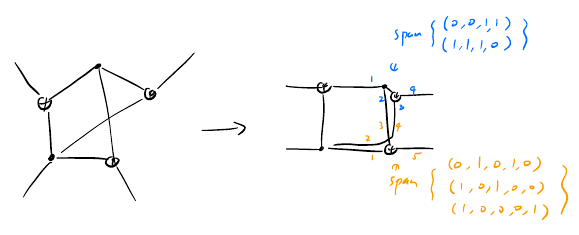}
  \caption{An example that transforms any geometrically local code into the family with a grid layout.}
  \label{fig:standardize-classical}
\end{figure}

\subsection{Quantum codes}

We describe the analogous family for 3D quantum local codes.
Let $Y$ be the cellular complex that subdivides $\RR^3$ using the integer lattice.
Each face is associated with a vector space $\cF_f$.
Each edge is associated with a vector subspace $\cF_e \subseteq \prod_{f \succ e} \cF_f$.
Each vertex is associated with a vector space $\cF_v \subseteq \prod_{f \succ v} \cF_f$,
  which consists of vectors ``consistent'' with $\cF_e$.
Specifically,
  $\cF_v = \{c \in \prod_{f \succ v} \cF_f: \forall e \in Y(1), \{c_f\}_{f \succ e} \in \cF_e\}$.
This naturally induces restriction maps between the cells $v \prec e \prec f$
  $\res_{v, e}, \res_{v, f}, \res_{e, f}$.
The quantum code is then induced from the chain complex,
  $\prod_{v \in Y(0)} \cF_v \xrightarrow{\delta_0} \prod_{e \in Y(1)} \cF_e \xrightarrow{\delta_1} \prod_{f \in Y(2)} \cF_f$,
  where $\delta_0(c)(e) = \sum_{v \prec e} \res_{v, e} c(v)$
    and $\delta_1(c)(f) = \sum_{e \prec f} \res_{e, f} c(e)$.

To summarize, specifying the quantum code requires a 2D cellular complex
  that embeds into the desired geometry,
  a vector space $\cF_f$ for each face $f \in Y(2)$,
  and a vector subspace $\cF_e$ for each edge $e \in Y(1)$.
Notice that $\cF_v$ is uniquely specified by the above data.
Not all combinations result in useful quantum codes.
For example, it may be beneficial to impose additional conditions, such as local acyclicity \cite[App A and B]{lin2024transversal}.

We sketch how to convert any local quantum code into a code within the above family.
First, we draw a segment between every qubit and check that have nontrivial interactions.
Next, we create squares based on the method from \cite{li2024transform},
  which we review here.
Consider a X check $v_0$ and a Z check $v_2$.
Since each X check and Z check share an even number of qubits,
  we can pair the qubits.
  Let $v_1, v_1'$ be such a pair.
We then form a square with vertices $v_0, v_1, v_2, v_1'$
  and repeat this process for all pairs of X and Z checks,
    along with their corresponding qubit pairs.
This square complex captures key properties of the code.

We can now slightly deform the square complex by a constant distance
  so that every qubit and check lies in $Y_0$,
    the deformed segments are supported on $Y_1$,
    and the deformed squares are supported on $Y_2$.
Note that a deformed segment may contain multiple cells in $Y(1)$
  and a deformed square might contain multiple cells in $Y(2)$.
This allows us to define $\cF_f$ by the number of surfaces overlaying on a 2-cell.
Finally, the faces around an edge are constrained by the structure of the code,
  which specifies $\cF_e$.

Overall, this means that every 3D quantum CSS geometrically local code can be viewed as
  a Tanner code over the 2D skeleton on the integer lattice grid.
Another way to say this is that every 3D quantum CSS geometrically local code can be viewed as
  multiple 2D surface codes stacked along the squares of the lattice grid,
  with carefully chosen condensation data along the edges.
This perspective aligns with the framework of topological defect networks \cite{aasen2020topological,song2023topological}.

This strategy can be extended to codes with arbitrary locality constraints.
Consider the family of quantum CSS codes that embeds in a space $Z$.
This family is equivalent to the family of Tanner codes on a 2-skeleton
  of a cellular decomposition $Y$ of $Z$,
  where each cell in $Y$ has a constant size $O(1)$.

\section{Nonconstruction: Brell's code} \label{sec:non-construction-brell}

An attempt to construct self-correcting quantum code based on fractals
  has been considered by Brell \cite{brell2016proposal}.
It was conjectured to be self-correcting,
  but we believe it is not the case.
At the same time, we believe that Brell's construction provides a subsystem code with
  $[[n = \Theta(L^{2+2\epsilon}), k = 1, d = \Omega(L^{1+\epsilon_2}), \cE = \Omega(L^\epsilon_3)]]$
  which fits into the cube $[0,L]^3$
  for some $\epsilon, \epsilon_2, \epsilon_3 > 0$.

The key difference between Brell's code and our fractal based construction in \Cref{sec:construction-fractal}
is that Brell uses the map from 0-cells to 1-cells of the prefractal\footnote{
  In fact, Brell described the code as a chain complex from 0-cells to 1-cells to 2-cells.
  Nevertheless, the chain complex structure is ``equivalent''
    to another map induced from 0-cells to 1-cells of a prefractal.
  The meaning of ``equivalent'' is chain homotopy equivalent,
    with the addition requirement that such equivalence can be implemented in a sparse way.
},
whereas we take a more complicated linear map supported on the prefractal,
  with small set expansion in both directions.
Note that the direction from 0-cells to 1-cells a has small set expansion,
  but this is not the case for the direction from 1-cells to 0-cells,
  which contains loop-like configurations.

We illustrate our claim that Brell's code is not self-correcting through a simpler example.
We consider a hypergraph product between the codes illustrated on the left of \Cref{fig:3D-toric-code-with-holes}.
It can be thought of as a 3D surface code with holes illustrated on the right of \Cref{fig:3D-toric-code-with-holes}.
The similarity to Brell's code is that the Sierpiński carpet also contains many holes.
As we will see, these holes modify the evolution of the system.

\begin{figure}[H]
  \centering
  \includegraphics[width=0.4\linewidth]{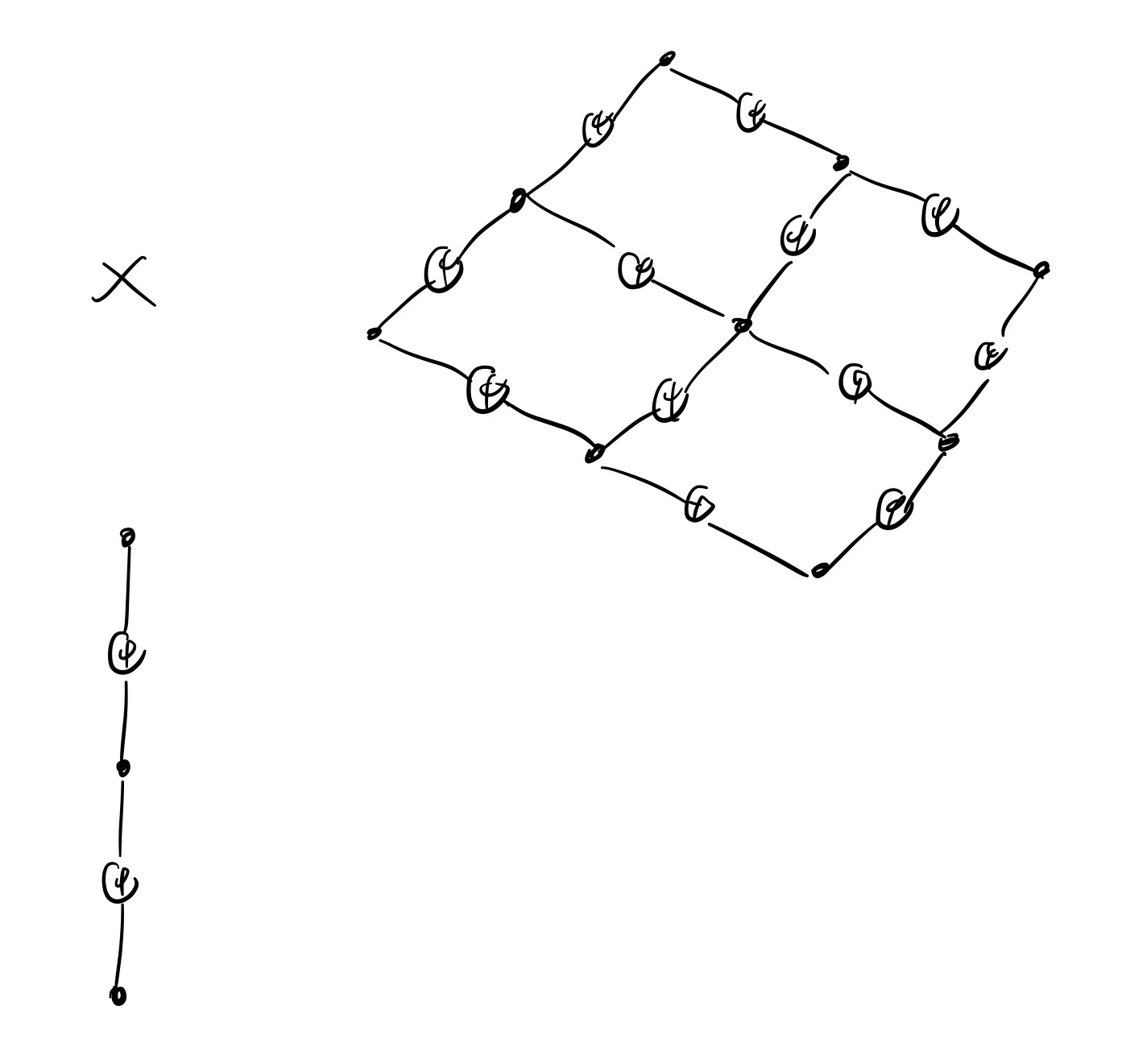}
  \hspace{1cm}
  \includegraphics[width=0.4\linewidth]{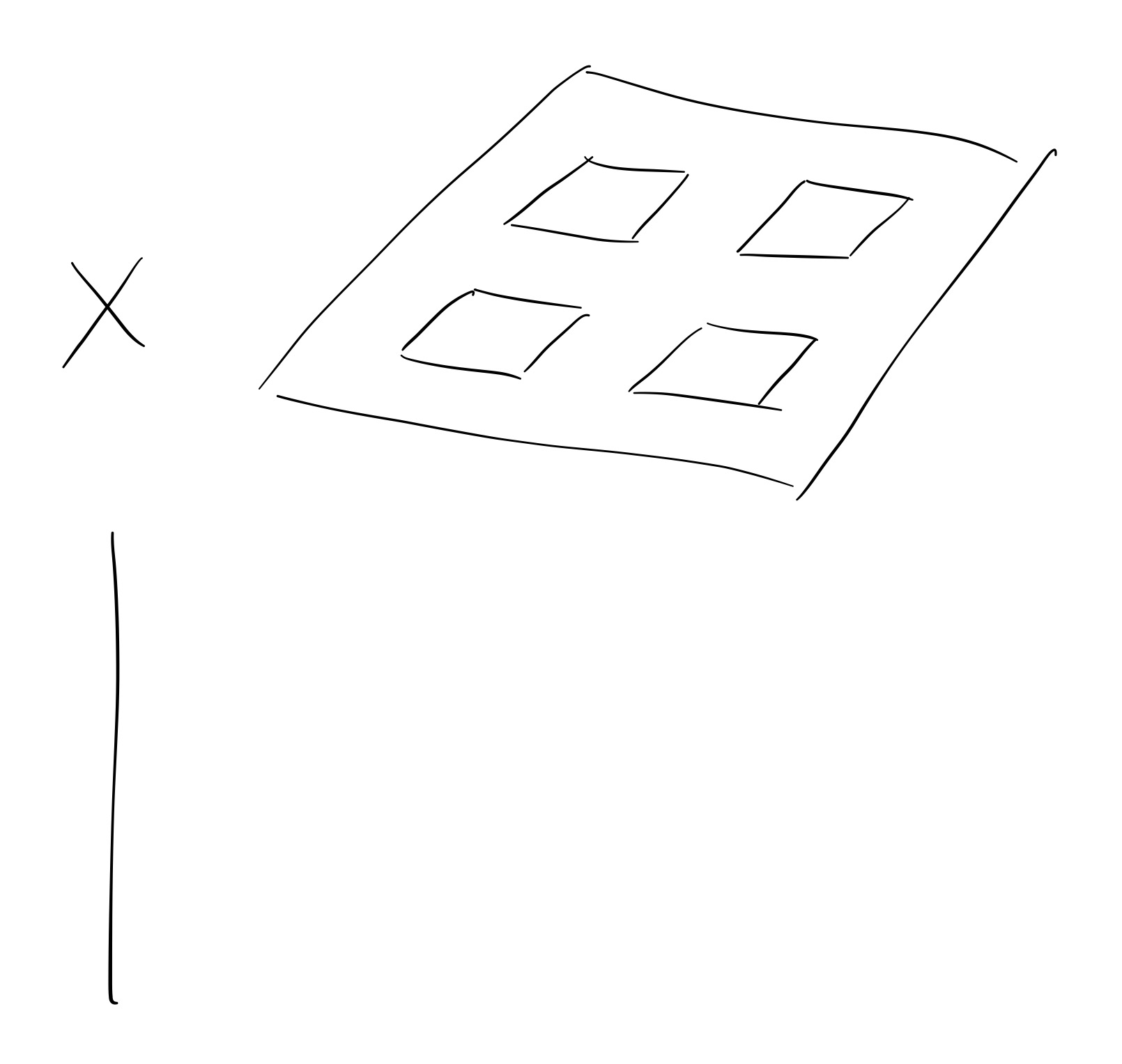}
  \caption{The figure shows a code that will be analyzed later.
            It can be viewed as a hypergraph product code between the two classical codes.
            Recall that to form a hypergraph product code,
              we put a qubit at vertices that corresponds to a product of two bits or two checks.
            This code can also be viewed as a 3D surface code over a manifold with holes.}
  \label{fig:3D-toric-code-with-holes}
\end{figure}

We first establish some basic facts of the code.
Because it is a hypergraph product code,
  it is clear that there is a nontrivial X-logical operator parallel to the plane
  as illustrated in \Cref{fig:3D-toric-code-with-holes-logical}.
Such an operator is analogous to the X-logical operator for a 3D surface code.
Assuming that all the holes have a constant size,
  this operator has distance $\Theta(L^2)$ and the energy barrier to implement it is $\Theta(L)$.
When the code is viewed as a subsystem code,
  these parameters are exactly the same as those for the traditional 3D surface code without holes.

\begin{figure}[H]
  \centering
  \includegraphics[width=0.38\linewidth]{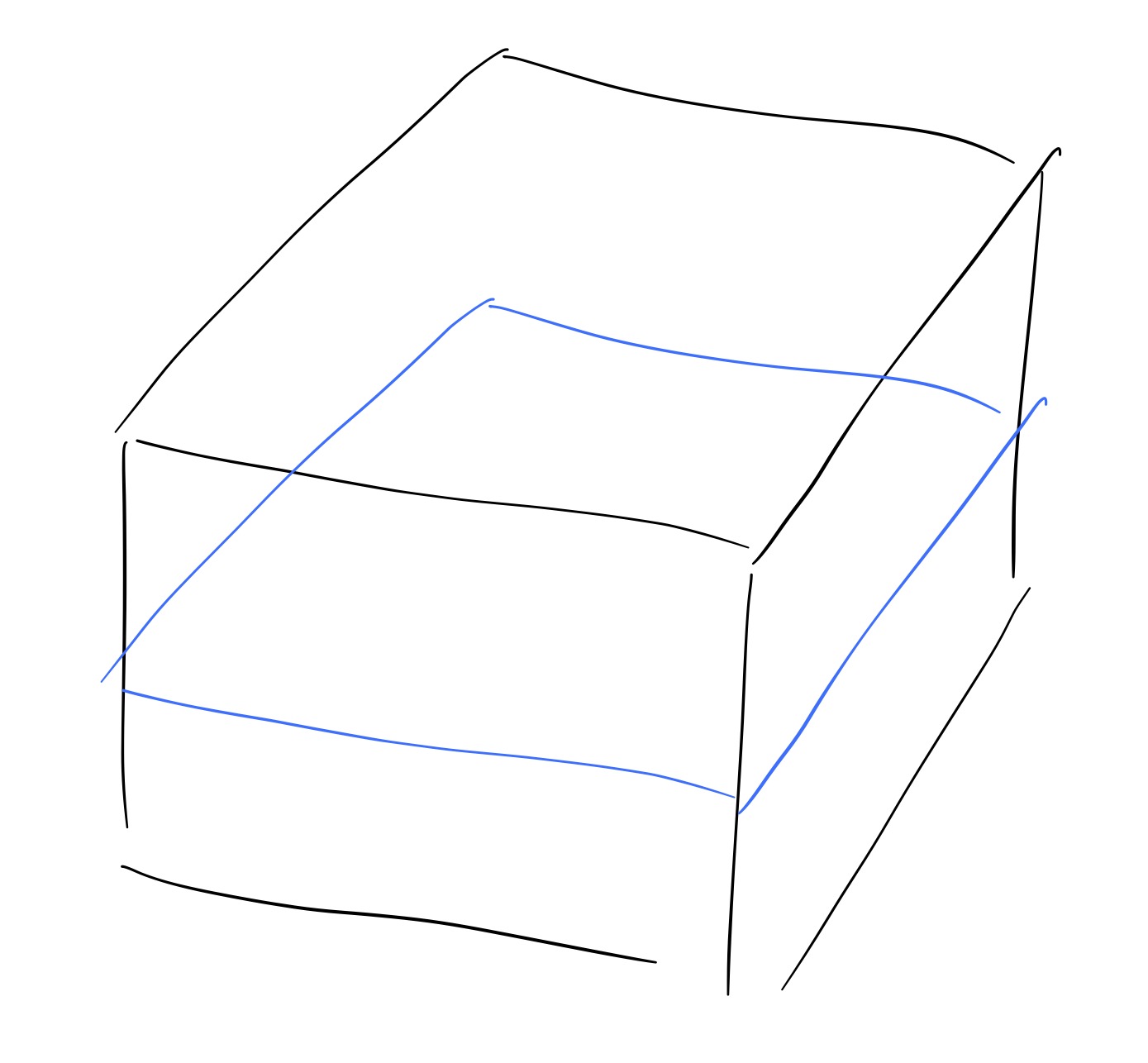}
  \hspace{1cm}
  \includegraphics[width=0.41\linewidth]{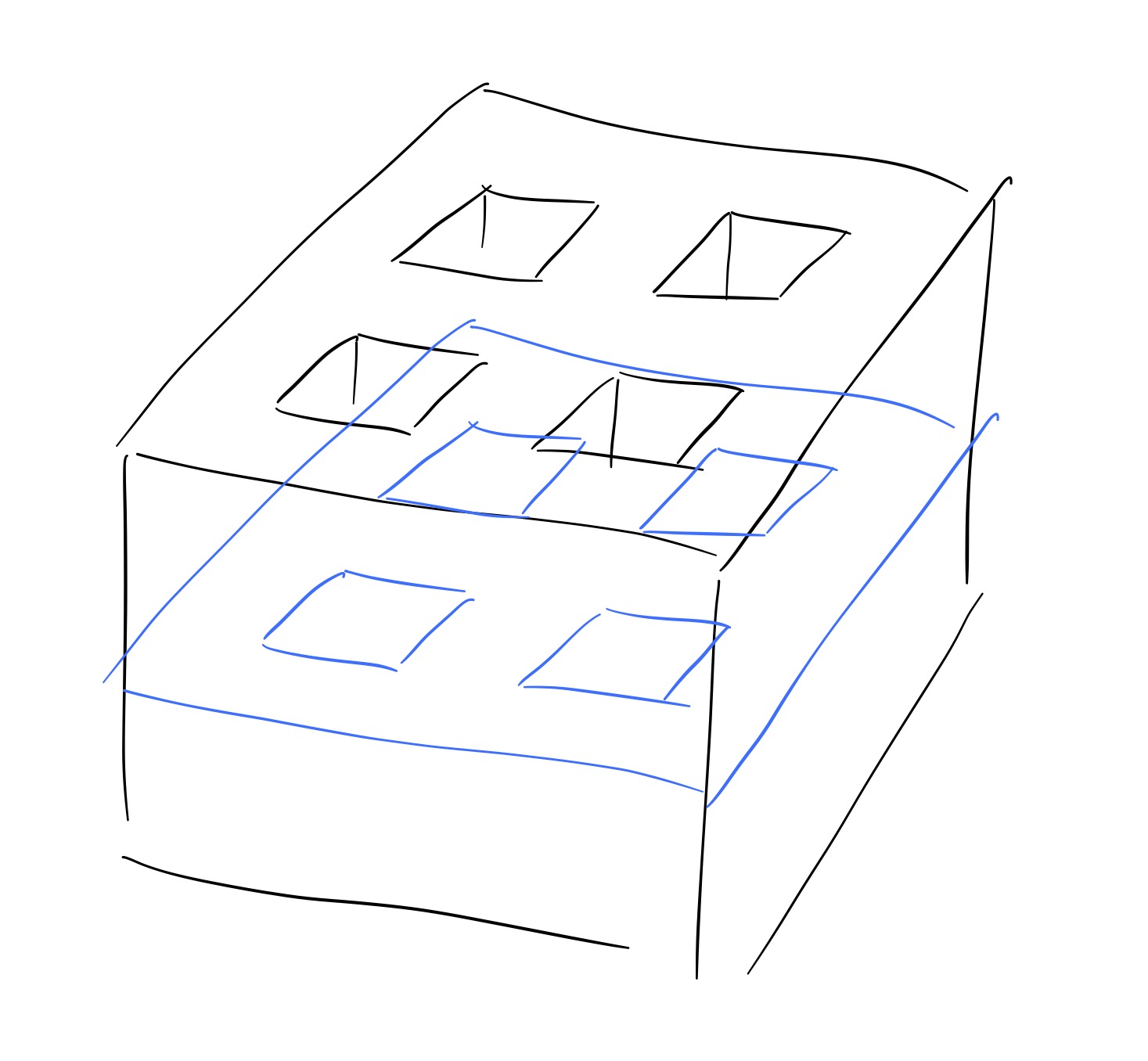}
  \caption{The figure shows the X-logical operator of the 3D surface code
            and the analogous X-logical operator of the code with holes.
            Both X-logical operators have a large distance and a large energy barrier.}
  \label{fig:3D-toric-code-with-holes-logical}
\end{figure}

It is known that such a X-logical operator is self-correcting for 3D surface codes.
The reason boils down to the fact that error patterns do not grow.
Instead, they want to contract and disappear.
In \Cref{fig:3D-toric-code-evolution},
  we see that if an error pattern grows,
  there is an associate energy penalty.
Recall, to read out the expectation of Z,
  assuming there is no error,
  the decoder draws a line from top to bottom
  and count the number of intersecting planes.
In the case with errors,
  because of the energy penalty,
  there will only be small loops.
That means to read out the expectation of Z,
  the decoder can again draw a line from top to bottom,
  with the additional process that detects and goes around small loops.

\begin{figure}[H]
  \centering
  \includegraphics[width=0.6\linewidth]{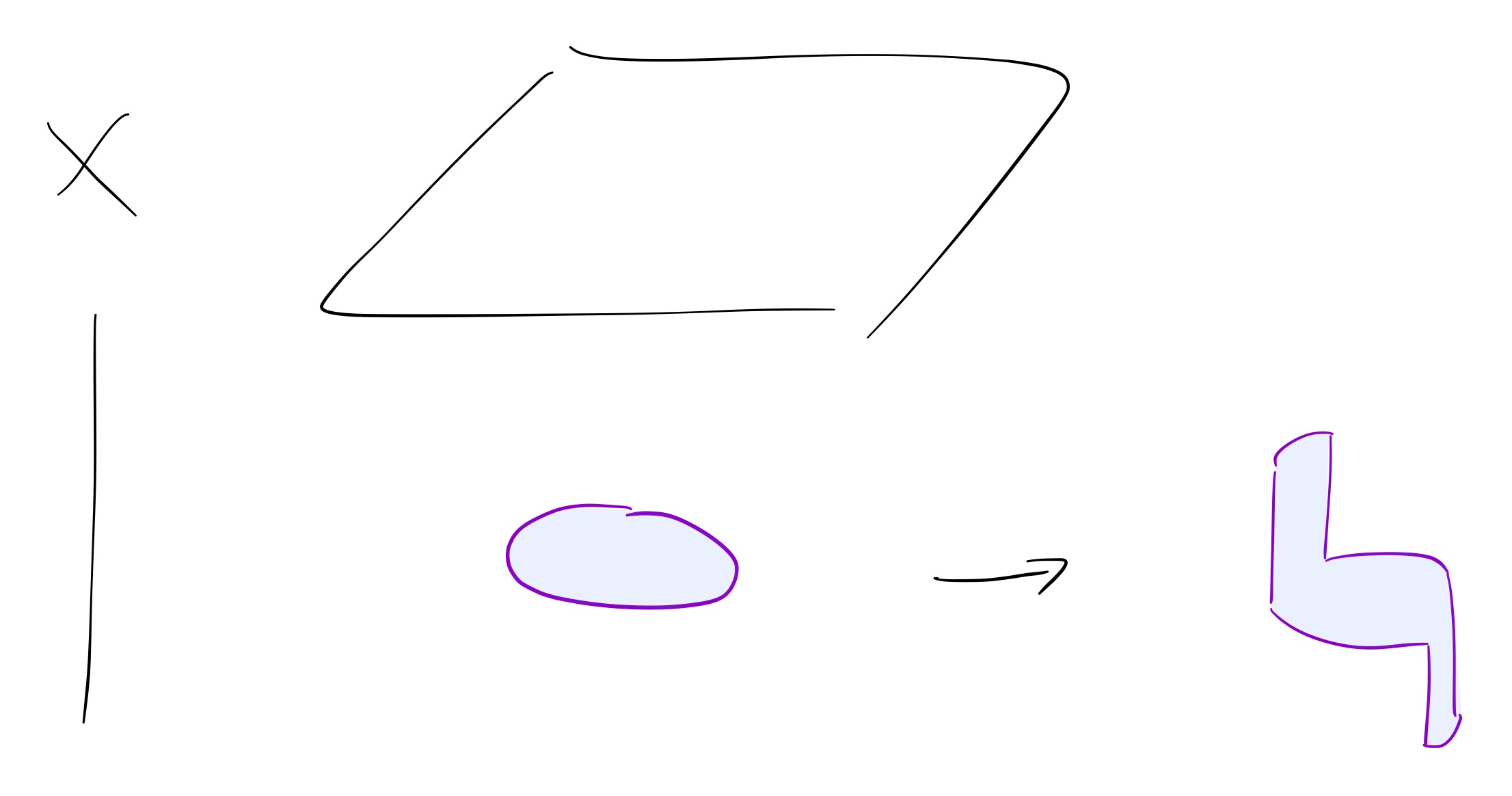}
  \caption{The error patterns are illustrated with the blue surfaces.
            The syndromes are illustrated with the purple lines.
            We see that the syndrome size increases when the error tries to grow.}
  \label{fig:3D-toric-code-evolution}
\end{figure}

However, the analogous X-logical operator for the code with holes is likely not self-correcting.
In \Cref{fig:3D-toric-code-with-holes-evolution},
  we see that an analogous error pattern can grow.
This is largely because the boundary condition around the holes
  makes the previous energy penalties associated with the vertical lines go away.
This causes the entropy effect to dominate,
  which prevents the error pattern from contracting and disappearing.
As a result, the decoder encounters difficulties
  in counting the number of planes
  intersecting a line from top to bottom.

\begin{figure}[H]
  \centering
  \includegraphics[width=0.6\linewidth]{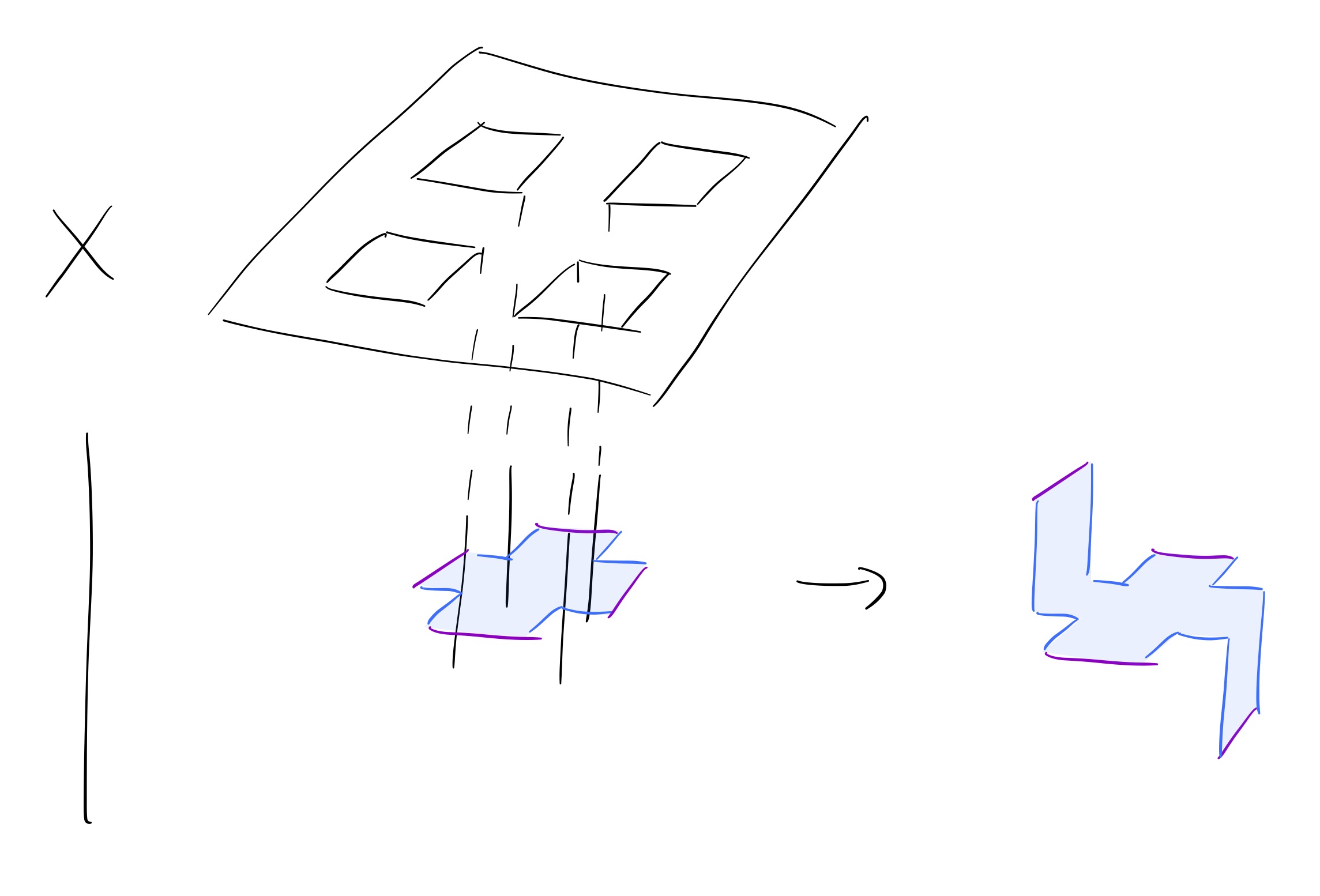}
  \caption{The error patterns are illustrated with the blue surfaces.
            The syndromes are illustrated with the purple lines.
            The boundary of the surface that does not contribute to the energy is illustrated with the blue lines.
            These blue lines are located on the boundary of the holes.
            We see that the syndrome size does not increase when the error grows in a certain way,
              even though the length of the boundary of the surface increases.}
  \label{fig:3D-toric-code-with-holes-evolution}
\end{figure}

It may be possible to add additional energy penalties at the boundaries
  to address issues about the holes.
Note that, however, this approach would go beyond the stabilizer model,
  and enter the realm of arbitrary 3D local Hamiltonian.

\section{Nonconstruction: 3D local code with linear energy barrier} \label{sec:non-construction-recent-codes}

Recently, there have been several constructions of 3D local quantum codes
  with optimal code dimension, distance, and energy barrier,
  $[[n = \Theta(L^3), k = \Theta(L), d = \Theta(L^2), \cE = \Theta(L)]]$
  \cite{portnoy2023local,lin2023geometrically,williamson2023layer,li2024transform}.
However, despite the large energy barrier,
  these codes likely do not have memory time $T_{mem} = \exp(\Theta(L))$.
Instead, they likely only have $T_{mem} \,\propto\, e^{a \beta - b}$
  for some constants $a, b$;
  the same behavior as a 2D toric code.

We use the construction in \cite{lin2023geometrically} to illustrate the point.
It is expected that other constructions have a similar behavior.
In \cite{lin2023geometrically},
  the code is obtained by gluing several 2D surface codes of size $L \times L$
  along their boundaries.
Its logical operators are tree-like with $\Theta(L)$ vertices
  whose vertices lies on the vertical boundaries
  and the edges runs across the surface codes.
See \Cref{fig:subdivided-code-logical-operator}.

\begin{figure}[H]
  \centering
  \includegraphics[width=0.3\linewidth]{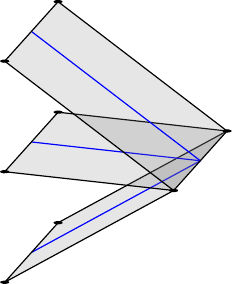}
  \caption{The figure shows part of the subdivided code described in \cite{lin2023geometrically} and one of its logical operators $\cL$.}
  \label{fig:subdivided-code-logical-operator}
\end{figure}

Since the anyon excitations in the 2D surface codes are mobile,
  after time $e^{\Theta(\beta)}$ they appear, diffuse, and spread across the surface
  with constant density of anyons (i.e. syndromes) everywhere.
See \Cref{fig:anyon-diffuse}.

\begin{figure}[H]
  \centering
  \includegraphics[width=0.3\linewidth]{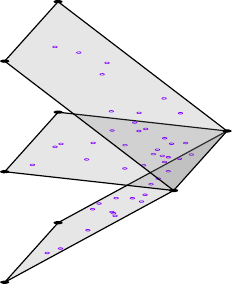}
  \caption{The figure shows how the system looks like after $e^{\Theta(\beta)}$ unit time.
            The anyons are represented by the purple circles.
            This is what the decode sees when it tries to recover the information.}
  \label{fig:anyon-diffuse}
\end{figure}

After time $e^{\Theta(\beta)}$, we suggest that it becomes impossible to distinguish between
  the case with initial state $|0\>$ over the case with initial state $\cL |0\>$,
  obtained by applying a logical operator $\cL$.
The reason is that for typical anyon configurations,
  the anyons have proliferated so much
  that there are multiple ways to connect them.
These ways may differ by a logical operator $\cL$,
  which means $|0\>$ becomes indistinguishable from $\cL |0\>$.
See \Cref{fig:indistinguishable-states}.

\begin{figure}[H]
  \centering
  \includegraphics[width=0.8\linewidth]{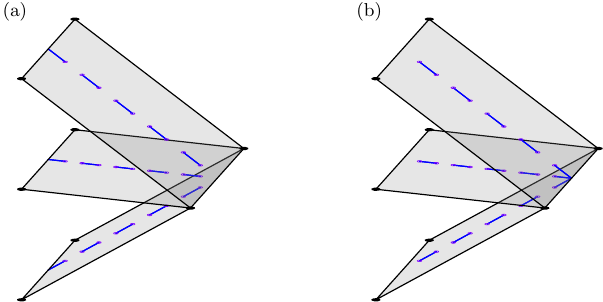}
  \caption{The figure shows two ways to decode a typical anyon configuration,
            which differ by a logical operator $\cL$.
            In general, there are more anyons.
            We only draw the relevant anyons to illustrate our point.}
  \label{fig:indistinguishable-states}
\end{figure}

\section{Discussions}

The paper discusses two approaches for constructing 3D self-correcting quantum stabilizer codes.
A related open problem is the construction of 4D self-correcting codes that allow transversal non-Clifford gates.
The known lower bound is 4D \cite{pastawski2015fault}
  and the current construction achieve this in 6D,
  using 6D color codes with transveral T gates \cite{bombin2013self}
  or 6D toric codes with transversal CCZ gates.
Based on the ideas presented in this work,
  by taking a product of three classical codes supported on a fractal with Hausdorff dimension $1 + \epsilon$,
  we suspect that it is possible to saturate the 4D lower bound,
  i.e. there exist 4D self-correcting CSS codes with transveral non-Clifford gates.

If we broaden the scope and relax the requirement that the Hamiltonian must be induced from a stabilizer code,
  the possibilities opens significantly.
Currently, there is no known lower bound on the dimensionality of self-correcting codes
  nor for self-correcting schemes that support universal quantum computation.
A recent work discovered a nonabelian TQFT beyond $\ZZ/2\ZZ$-gauge theory
  leading to a 5D self-correcting code with non-Clifford gates \cite{hsin2024non}.
It may also be worthwhile to explore quantum codes based on
  conformal field theory (CFT) \cite{sang2024approximate} or holographic CFT \cite{kibe2022holographic}.
CFTs offer several desireable properties,
  including the local indistinguishability of low energy states,
  and they allow different types of analysis using tools from quantum field theory.

\section*{Acknowledgements}
TCL thanks Wilbur Shirley and Jeongwan Haah for the discussion on Theorem 4.3 in \cite{haah2013lattice}.
TCL was supported in part by funds provided by the U.S. Department of Energy (D.O.E.) under the cooperative research agreement DE-SC0009919 and by the Simons Collaboration on Ultra-Quantum Matter, which is a grant from the Simons Foundation (652264 JM).

\bibliographystyle{unsrt}
\bibliography{references.bib}

\end{document}